\documentclass[a4paper,11pt]{article}
\usepackage{amsthm}
\usepackage{algorithm2e}
\usepackage{color}
\usepackage{graphicx}
\usepackage{makeidx}
\usepackage{amsmath}
\usepackage{amssymb}
\usepackage{float}
\usepackage{epsfig}

\begin{document}

\newcommand{\MinCut}[1]{\ensuremath{\mathit{MinCut}\left({#1}\right)}\xspace}
\newcommand{\DynMinCut}[1]{\ensuremath{\mathit{DynMinCut}\left({#1}\right)}\xspace}
\newcommand{\Cut}[1]{\ensuremath{\mathit{Cut}\left({#1}\right)}\xspace}
\newcommand{\Dyn}[1]{\ensuremath{\mathit{Dyn}\left({#1}\right)}\xspace}

\newtheorem{theorem}{Theorem}
\newtheorem{lemma}{Lemma}
\newtheorem{definition}{Definition}

\title{Reliable Communication in a Dynamic Network\\
in the Presence of Byzantine Faults}
\author{Alexandre Maurer$^1$,
S\'{e}bastien Tixeuil$^{2,3}$ and
Xavier Defago$^3$\\
$^1$ \'{E}cole Polytechnique F\'{e}d\'{e}rale de Lausanne\\
$^2$ Sorbonne Universit\'{e}s, UPMC Univ. Paris 06, LIP6 CNRS UMR 7606\\
$^3$ Institut Universitaire de France\\
$^4$ Japan Advanced Institute of Science and Technology (JAIST)\\
E-mail: Alexandre.Maurer@epfl.ch,
Sebastien.Tixeuil@lip6.fr, Defago@jaist.ac.jp}

\maketitle

\begin{abstract}
We consider the following problem: two nodes want to reliably communicate in a dynamic multihop network where some nodes have been compromised, and may have a totally arbitrary and unpredictable behavior. These nodes are called \emph{Byzantine}. We consider the two cases where cryptography is available and not available.

We prove the necessary and sufficient condition (that is, the weakest possible condition) to ensure reliable communication in this context. Our proof is constructive, as we provide Byzantine-resilient algorithms for reliable communication that are optimal with respect to our impossibility results.

In a second part, we investigate the impact of our conditions in three case studies: participants interacting in a conference, robots moving on a grid and agents in the subway. Our simulations indicate a clear benefit of using our algorithms for reliable communication in those contexts.
\end{abstract}

\section{Introduction}

As modern networks grow larger, their components become more likely to fail, sometimes in unforeseen ways. As opportunistic networks become more widespread, the lack of global control over individual participants makes those networks particularly vulnerable to attacks. Many failure and attack models have been proposed, but one of the most general is the \emph{Byzantine} model proposed by Lamport et al.~\cite{LSP82j}. The model assumes that faulty nodes can behave arbitrarily. In this paper, we study the problem of reliable communication in a multihop network despite the presence of Byzantine faults. The problem proves difficult since even a single Byzantine node, if not neutralized, can lie to the entire network.

\subsection*{Related works}

A common way to solve this problem is to use \emph{cryptography} \cite{CL99c,DFS05c}: the nodes use digital signatures to authenticate the sender across multiple hops.
However, cryptography \emph{per se} is not unconditionally reliable, as shown by the recent Heartbleed bug~\cite{hbleed} discovered in the widely deployed OpenSSL software. The \emph{defense in depth} paradigm~\cite{indepth} advocates the use of multiple layers of security controls, including non-cryptographic ones. For instance, if the cryptography-based security layer is compromised by a bug, a virus, or intentional tampering, a cryptography-free communication layer can be used to safely broadcast a patch or to update cryptographic keys. 
Thus, it is interesting to develop both cryptographic and
non-cryptographic strategies.

Following the setting of the seminal paper of Lamport et al.~\cite{LSP82j}, many subsequent papers focusing of Byzantine tolerance~\cite{AW98b,MMR03j,MRRS01c,MS03j} study agreement and reliable communication primitives using cryptography-free protocols in networks that are both \emph{static} and \emph{fully connected}. A recent exception to fully connected topologies in Byzantine agreement protocols is the recent work of Tseng, Vaidya and Liang~\cite{TV13,VTL12}, which considers specific classes of \emph{static} directed graphs (\emph{i.e.}, graphs with a particularly high clustering coefficient) and considers \emph{approximate} and \emph{iterative} versions of the agreement problem.

In general multihop networks, two notable classes of algorithms use some locality property to tolerate Byzantine faults: space-local and time-local algorithms. Space-local algorithms~\cite{MT07j,NA02c,SOM05c} try to contain the fault (or its effect) as close to its source as possible. This is useful for problems where information from remote nodes is unimportant (such as vertex coloring, link coloring, or dining philosophers). Time-local algorithms~\cite{DMT10ca,DMT10cd,DMT11j,DMT11cb,MT06cb} try to limit over time the effect of Byzantine faults. Time-local algorithms presented so far can tolerate the presence of at most a single Byzantine node, and are unable to mask the effect of Byzantine actions. Thus, neither approach is suitable to reliable communication.

In dense multihop networks, a first line of work assumes that there is a bound on the fraction of Byzantine nodes among the neighbors of each node. Protocols have been proposed for nodes organized on a grid \cite{BV05c,K04c} (but with much more than $4$ neighbors), and later generalized to other topologies \cite{PP05j}, with the assumption that each node knows the global topology. Since this approach requires all nodes to have a large degree, it may not be suitable for every multihop networks. The case of sparse networks was studied under the assumption that Byzantine failures occur uniformly at random~\cite{CtrZ,Scalbyz,Trig}, an assumption that holds, \emph{e.g.}, in structured overlay networks where the identifier (\emph{a.k.a.} position) of a new node joining the network is assigned randomly, but not necessarily in various actual communication networks.

Most related to our work is the line of research that assume the existence of $2k+1$ node-disjoint paths from source to destination, in order to provide reliable communication in the presence of up to $k$ Byzantine failure~\cite{D82j,NT09j,secure_mt}. The initial solution~\cite{D82j} assumes that each node is aware of the global network topology, but this hypothesis was dropped in subsequent work~\cite{NT09j,explorratum}.

None of the aforementioned papers considers genuinely dynamic networks, \emph{i.e.}, where the topology evolves while the protocol executes.

\begin{figure}
\begin{center}
\includegraphics[width=8cm]{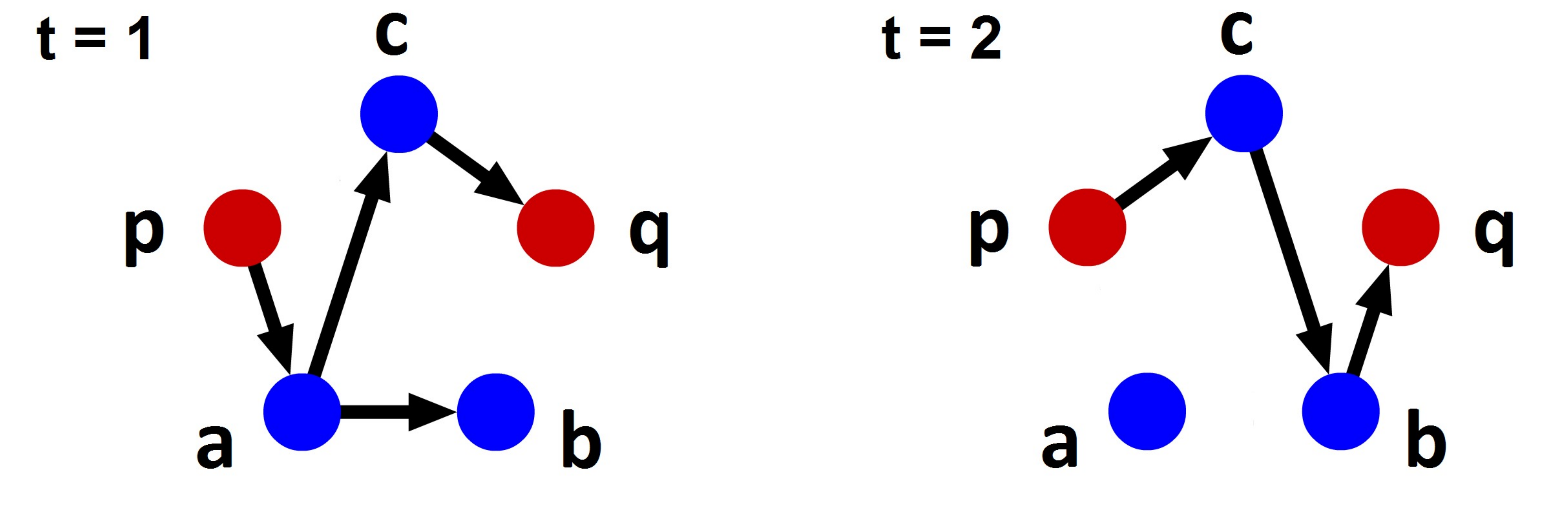}
\caption{Counterexample to Menger's theorem in dynamic graphs.} 
\label{fig:menger}
\end{center}
\end{figure}

\subsection*{Our contribution}

In this paper, our objective is to determine the condition for reliable communication in the presence of up to $k$ Byzantine failures in a \emph{dynamic} network, where the topology can vary with time. The proof technique used in \cite{D82j,NT09j,secure_mt} implicitly relies on Menger's theorem \cite{menger_thm}, which can be expressed as follows: there exists $x$ disjoint paths between two nodes $p$ and $q$ if and only if $x$ nodes must removed to disconnect
$p$ and $q$.

However, Menger's theorem does not generalize to dynamic networks~\cite{menger_fail}.
To illustrate this, let us consider the simple dynamic network of Figure~\ref{fig:menger}. This network is in two steps ($t = 1$ and $t = 2$). There exists three dynamic paths connecting $p$ to $q$: $(p,a,c,q)$, $(p,c,b,q)$ and $(p,a,b,q)$. To cut these three paths, at least two nodes must be removed: either $\{a,b\}$, $\{b,c\}$ or $\{a,c\}$. Yet, it is impossible to find two disjoint paths among the three dynamic paths.
Therefore, Menger's theorem cannot be used to prove the condition in dynamic networks.

In this paper, we prove the necessary and sufficient condition for reliable communication in dynamic networks, in the presence of up to $k$ Byzantine failures.
We consider the two cases where cryptography is available and not available. Our characterization is based on a dynamic version of a minimal cut between $p$ and $q$, denoted by DynMinCut($p,q$), that takes into account both the presence of particular paths and their duration with respect to the delay that is necessary to actually transmit a message over a path. Then condition is that DynMinCut($p,q$) is lower or equal to $2k$ (without cryptography) or $k$ (with cryptography). The proof is constructive, as we provide algorithms to prove the sufficiency of the condition.

In a second part, we apply these conditions to three case studies: participants interacting in a conference, robots moving on a grid and agents moving in the subway. We thus show the benefit of this multihop approach for reliable communication, instead of waiting that the source meets the sink directly (if this event is to occur).

\subsection*{Organization of the paper}

The paper is organized as follows.
In Section~\ref{sec_prelim}, we present the model and give basic definitions. In Section~\ref{sec_algo} (resp. \ref{sec_algo2}), we give the algorithm and prove the condition for the non-cryptographic (resp. cryptographic) case. We present the case studies in Section~\ref{sec_case}.

\section{Preliminaries}
\label{sec_prelim}


\subsection*{Network model}

We consider a continuous temporal domain $\mathbb{R}^{+}$ where dates are positive real numbers. We model the system as a time varying graph, as defined by Casteigts, Flocchini, Quattrociocchi and Santoro~\cite{dynam_model}, where vertices represent the processes and edges represent the communication links (or channels). A time varying graph is a dynamic graph represented by a tuple $\mathcal{G} = (V,E,\rho,\zeta)$ where:
\begin{itemize}
\item $V$ is the set of \emph{nodes}.
\item $E \subseteq V \times V$ is the set of \emph{edges}.
\item $\rho : E \times \mathbb{R}^{+} \rightarrow \{0,1\}$
	is the \emph{presence} function: $\rho(e,t) = 1$ indicates that edge~$e$ is present at date~$t$.
\item $\zeta : E \times \mathbb{R}^{+} \rightarrow \mathbb{R}^{+}$
	is the \emph{latency} function: $\zeta(e,t) = T$ indicates that a message sent at date $t$ takes $T$ time units to cross edge $e$.
\end{itemize}

The discrete time model is a special case, where time and latency are restricted to integer values.

\subsection*{Hypotheses}

We make the same hypotheses as previous work on the subject \cite{BV05c,D82j,K04c,CtrZ,Trig,Scalbyz,NT09j,PP05j}. First, each node has a unique identifier. Then, we assume \emph{authenticated channels} (or \emph{oral model}), that is, when a node $q$ receives a message through channel $(p,q)$, it knows the identity of $p$. Now, an omniscient adversary can select up to $k$ nodes as \emph{Byzantine}. These nodes can have a totally arbitrary and unpredictable behavior defined by the adversary (including tampering or dropping messages, or simply crashing). Finally, other nodes are \emph{correct} and behave as specified by the algorithm. Of course, correct nodes are unable to know \emph{a priori} which nodes are Byzantine. We also assume that a correct node $u$ is aware of its \emph{local topology} at any given date $t$ (that is, $u$ knows the set of nodes $v$ such that $\rho((u,v),t) = 1$).

\subsection*{Dynamicity-related definitions}

Informally, a \emph{dynamic path} is a sequence of nodes a message can traverse, with respect to network dynamicity and latency.

\begin{definition} [Dynamic path]
\label{def_dyn}
A sequence of distinct nodes $(u_1,\dots,u_n)$ is a \emph{dynamic path} from $u_1$ to $u_n$ if and only if there exists a sequence of dates $(t_1,\dots,t_{n})$ such that, $\forall i \in \{1,\dots,n-1\}$ we have:
\begin{itemize}
\item $e_i = (u_i,u_{i+1}) \in E$, i.e. there exists an edge connecting $u_i$ to $u_{i+1}$.
\item $\forall t \in [t_i, t_i + \zeta(e_i,t_i)]$, $\rho(e_i,t) = 1$, i.e. $u_i$ can send a message to $u_{i+1}$ at date $t_i$.
\item $\zeta(e_i,t_i) \leq t_{i+1} - t_i$, i.e. the aforementioned message is received by date $t_{i+1}$.
\end{itemize}
\end{definition}

We now define the \emph{dynamic minimal cut} between two nodes $p$ and $q$ as the minimal number of nodes (besides $p$ and $q$) one has to remove from the network to prevent the existence of a dynamic path between $p$ and $q$. Formally: 
\begin{itemize}
\item Let \Dyn{p,q} be the set of node sets $\{u_1,\dots,u_n\}$ 
  such that $(p,u_1,\dots,u_n,q)$ is a dynamic path.
\item For a set of node sets $\Omega = \{S_1,\dots,S_n\}$,
let \Cut{\Omega} be the set of node sets $C$ such that,
$\forall i \in \{1,\dots,n\}$,
$C \cap S_i \neq \emptyset$ ($C$ contains at least one node from each set $S_i$).
\item Let $\MinCut{\Omega} = \min_{C 
\Cut{\Omega}}
|C|$ (the size of the smallest element of \Cut{\Omega}). If $Cut(\Omega)$ is empty, we assume that $\MinCut{\Omega} = +\infty$.
\item Let $\DynMinCut{p,q} = \MinCut{\Dyn{p,q}}$.
\end{itemize}

\subsection*{Problem specification}

We say that a node \emph{multicasts} a message $m$ when it sends $m$ to all nodes in its current local topology. Now, a node $u$ \emph{accepts} a message $m$ from another node $v$ when it considers that $v$ is the author of this message. We now define our problem specification, that is, \emph{reliable} communication.

\begin{definition} [Reliable communication]
\label{def_rel}
Let $p$ and $q$ be two correct nodes. An algorithm ensures \emph{reliable communication}
from $p$ to $q$ when the following two conditions are satisfied:
\begin{itemize}
\item When $q$ accepts a message from $p$, $p$ is necessarily the author of this message.
\item When $p$ sends a message, $q$ eventually receives and accepts this message from $p$.
\end{itemize}

\end{definition}

\section{Non-cryptographic reliable communication}
\label{sec_algo}

In this section, we consider that cryptography is not available. We first provide a Byzantine-resilient multihop broadcast protocol. This algorithm is used as a constructive proof for the sufficient condition for reliable communication.
We then prove the necessary and sufficient condition for reliable communication.

\subsection*{Informal description of the algorithm}

Consider that each correct node $p$ wants to broadcast a message $m_0$ to the rest of the network. Let us first discuss why the naive flood-based solution fails. A naive first idea would be to send a tuple $(p,m_0)$ through all possible dynamic paths: thus, each node receiving $m_0$ knows that $p$ broadcast $m_0$. Yet, Byzantine nodes may forward false messages, \emph{e.g.}, a Byzantine node could forward the tuple $(p,m_1)$, with $m_1 \neq m_0$, to make the rest of the network believe that $p$ broadcast $m_1$.

To prevent correct nodes from accepting false message, we attach to each message the set of nodes that have been visited by this message since it was sent (that is, we use $(p,m,S)$, where $S$ is a set of nodes already visited by $m$ since $p$ sent it). 
As the Byzantine nodes can send any message, in particular, they can forward false tuples $(p,m,S)$.
Therefore, a correct node only accepts a message when it has been received through a collection of dynamic paths that cannot be cut by $k$ nodes (where $k$ is a parameter of the algorithm, and supposed to be an upper bound on the total number of Byzantine nodes in the network).

\subsection*{Variables}

Each correct node $u$ maintains the following variables:
\begin{itemize}
\item $u.m_0$, the message that $u$ wants to broadcast.
\item $u.\Omega$, a dynamic set registering all tuples $(s,m,S)$ received by $u$.
\item $u.Acc$, a dynamic set of confirmed tuples $(s,m)$.
We assume that whenever $(s,m) \in u.Acc$, $u$ accepts $m$ from $s$.
\end{itemize}
Initially,  $u.\Omega = \{(u,u.m_0,\o)\}$ and $u.Acc = \{(u,u.m_0)\}$.

\subsection*{Algorithm}

Each correct node $u$ obeys the three following rules:
\begin{enumerate}
\item Initially, and whenever $u.\Omega$ or the local topology of $u$ change: multicast $u.\Omega$.
\item Upon reception of $\Omega'$ through channel $(v,u)$:
$\forall (s,m,S) \in \Omega'$, if $v \notin S$ then append $(s,m,S \cup \{v\})$ to $u.\Omega$.
\item Whenever there exist $s$, $m$ and $\{S_1, \dots, S_n\}$ such that $\forall i \in \{1,\dots,n\}$, $(s,m,S_i \cup \{s\}) \in u.\Omega$
and $\MinCut{\{S_1,\dots,S_n\}} > k$: append $(s,m)$ to $u.Acc$.
\end{enumerate}

\subsection*{Condition for reliable communication}
\label{genproof}

Let us consider a given dynamic graph, and two given correct nodes $p$ and $q$. Our main result is as follows:

\begin{theorem}
\label{mainthm}
For a given dynamic graph, 
a $k$-Byzantine tolerant reliable communication from $p$ to $q$ is feasible if and only if $\DynMinCut{p,q} > 2k$.
\end{theorem}

\begin{proof}
The proof of the ``only if'' part is in Lemma~\ref{lemnec}. The proof of the ``if''
is in Lemma~\ref{lemsuf}.
\end{proof}

\begin{lemma}[Necessary condition]
\label{lemnec}
For a given dynamic graph,
let us suppose that there exists an algorithm ensuring
reliable communication from $p$ to $q$.
Then, we necessarily have $\DynMinCut{p,q} > 2k$. 
\end{lemma}

\begin{proof}
Let us suppose the opposite: 
there exists an algorithm ensuring
reliable communication from $p$ to $q$, and yet, $\DynMinCut{p,q} \leq 2k$. Let us show that it leads to a contradiction.

As we have $\DynMinCut{p,q} = \MinCut{\Dyn{p,q}} \leq 2k$
and $\MinCut{\Dyn{p,q}} = \min_{C \in \Cut{\Dyn{p,q}}}$
$|C|$,
there exists an element $C$ of \Cut{\Dyn{p,q}} such that
$|C| \leq 2k$.
Let $C_1$ be a subset of $C$ containing $k'$ elements, with $k' = \min(k,|C|)$. Let $C_2 = C - C_1$. Thus, we have
$|C_1| \leq k$ and $|C_2| \leq k$.

According to the definition of \Cut{\Dyn{p,q}},
 $C$ contains a node of each possible dynamic path from $p$ to $q$. Therefore, the information that $q$ receives about $p$ are completely determined by the behavior of the nodes in $C$.
 
Let us consider two possible placements of Byzantine nodes, and show that they lead to a contradiction:

\begin{itemize}
\item First, suppose that all nodes in $C_1$ are Byzantine, and that all other nodes are correct. This is possible since  $|C_1| \leq k$.

Suppose now that $p$ broadcasts a message $m$.
Then, according to our hypothesis, since the algorithm ensures reliable communication, $q$ eventually accepts $m$ from $p$,
regardless of what the behavior of the nodes in $C_1$ may be.

\item Now, suppose that all nodes in $C_2$ are Byzantine, and that all other nodes are correct. This is also possible since $|C_2| \leq k$.

Then, suppose that $p$ broadcasts a message $m' \neq m$, and that the Byzantine nodes have exactly the same behavior
as the nodes of $C_2$ had in the previous case.

Thus, as the information that $q$ receives about $p$ is completely determined by the behavior of the nodes of $C$, from the point of view of $q$, this situation is indistinguishable from the previous one:
the nodes of $C_2$ have the same behavior, and the behavior of the nodes of $C_1$ is unimportant.
Thus, similarly to the previous case, $q$ eventually accepts $m$ from $p$.
\end{itemize}

Therefore, in the second situation, $p$ broadcasts $m$, and $q$ eventually accepts $m' \neq m$. Thus,
according to Definition~\ref{def_rel},
the algorithm does not ensure reliable communication,  which contradicts our initial hypothesis. Hence, the result.
\end{proof}

\begin{lemma}[Safety]
\label{lemsafe}
Let us suppose that all correct nodes follow our algorithm. If $(p,m) \in q.Acc$, then $m = p.m_0$.
\end{lemma}

\begin{proof} As $(p,m) \in q.Acc$, according to rule $3$ of our algorithm, there exists $\{S_1, \dots, S_n\}$ such that, $\forall i \in \{1,\dots,n\}$, $(p,m,S_i \cup \{p\}) \in q.\Omega$, and $\MinCut{\{S_1,\dots,S_n\}} > k$.

Suppose that each node set $S \in \{S_1,\dots,S_n\}$ contains at least one Byzantine node. If $C$ is the set of Byzantine nodes, then $C \in \Cut{\{S_1,\dots,S_n\}}$ and $|C| \leq k$. This is impossible because $\MinCut{\{S_1,\dots,S_n\}} > k$. Therefore, there exists $S \in \{S_1,\dots,S_n\}$ such that $S$ does not contain any Byzantine node.

Now, let us use the correct dynamic path corresponding to $S$ to show that $m = m_0$. Let $n' = |S \cup \{p\}|$. Let us show the following property $\mathcal{P}_i$ by induction,
$\forall i \in \{0, \dots, n'\}$: there exists a correct node $u_i$ and a set of correct nodes $X_i$ such that $(p,m,X_i) \in u_i.\Omega$ and $|X_i| = |S \cup \{p\}| - i$.

\begin{itemize}

\item As $S \in \{S_1,\dots,S_n\}$, $(p,m,S \cup \{p\}) \in q.\Omega$. Thus, $\mathcal{P}_0$ is true if we take $u_0 = q$ and $X_0 = S \cup \{p\}$.

\item Let us now suppose that $\mathcal{P}_{i+1}$ is true, for $i < n'$. As $(p,m,X_i) \in u_i.\Omega$, according to rule $2$ of our algorithm, it implies that $u_i$ received $\Omega'$ from a node $v$, with $(p,m,X) \in \Omega'$, $v \notin X$ and $X_i = X \cup \{v\}$. Thus,
$|X| = |X_i| - 1 = |S\cup \{p\}| - (i+1)$.

As $v \in X_i$ and $X_i$ is a set of correct nodes, $v$ is correct and behaves according to our algorithm. Then, as $v$ sent $\Omega'$, according to rule $1$ of our algorithm, we necessarily have $\Omega' \subseteq v.\Omega$. Thus, as $(p,m,X) \in \Omega'$, we have $(p,m,X) \in v.\Omega$. Hence, $\mathcal{P}_{i+1}$ is true if we take $u_{i+1} = v$ and $X_{i+1} = X$.
\end{itemize}

By induction principle, $\mathcal{P}_{n'}$ is true.
As $|X_{n'}|  =  0$, $X_{n'} = \o$ and $(p,m,\o) \in u_{n'}$. As $u_{n'}$ is a correct node and follows our algorithm, the only possibility to have $(p,m,\o) \in u_{n'}.\Omega$ is that $u_{n'} = p$ and $m = p.m_0$. Thus, the result.
\end{proof}

\begin{lemma}[Communication]
\label{lemrel}
Let us suppose that $\DynMinCut{p,q} > 2k$, and that all correct nodes follow our algorithm. Then, we eventually have $(p,p.m_0) \in q.Acc$.
\end{lemma}

\begin{proof}
Let $\{S_1,\dots,S_n\}$ be the set of node sets $S \in \Dyn{p,q}$ that only contain correct nodes. Similarly, let $\{X_1,\dots,X_{n'}\}$ be the set of node sets $X \in \Dyn{p,q}$ that contain at least one Byzantine node.

Let us suppose that $\MinCut{\{S_1,\dots,S_n\}} \leq k$. Then, there exists $C \in \Cut{\{S_1,\dots,S_n\}}$ such that $|C| \leq k$. Let $C'$ be the set containing the nodes of $C$ and the Byzantine nodes. Thus, and $C' \in \Cut{\{S_1,\dots,S_n\} \cup \{X_1,\dots,X_{n'}\}} = \Cut{\Dyn{p,q}}$, and $|C'| \leq 2k$. Thus, $\MinCut{\Dyn{p,q}} \leq 2k$, which contradicts our hypothesis. Therefore, $\MinCut{\{S_1,\dots,S_n\}} > k$.

In the following, we show that $\forall S \in \{S_1,\dots,S_n\}$, we eventually have $(p,p.m_0,S \cup \{p\}) \in q.\Omega$, ensuring that $q$ eventually accepts $p.m_0$ from $p$.

Let $S \in \{S_1,\dots,S_n\}$. As $S \in \Dyn{p,q}$, let $(u_1,\dots,u_N)$ be the dynamic path such that $p = u_1$, $q = u_N$ and $S = \{u_2,\dots,u_{N-1}\}$. Let $(t_1,\dots,t_N)$ be the corresponding dates, according to Definition~\ref{def_dyn}. Let us show the following property $\mathcal{P}_i$ by induction, $\forall i \in \{1,\dots,N\}$: at date $t_i$, $(p,p.m_0,X_i) \in u_i.\Omega$, with $X_i = \o$ if $i = 1$ and $\{u_1,\dots,u_{i-1}\}$ otherwise.

\begin{itemize}
\item $\mathcal{P}_1$ is true, as we initially have $(p,p.m_0,\o) \in p.\Omega$.
\item Let us suppose that $\mathcal{P}_i$ is true, for $i < N$. According to Definition~\ref{def_dyn}, $\forall t \in [t_i,t_i + \zeta(t_i,u_i)]$, $\rho(e_i,t) = 1$, $e_i$ being the edge connecting $u_i$ to $u_{i+1}$.

\begin{itemize}
\item Let $t_A \leq t_i$ be the earliest date such that, $\forall t \in [t_A,t_i + \zeta(t_i,u_i)]$, $\rho(e_i,t) = 1$. 
\item Let $t_B \leq t_i$ be the date where $(p,m,X_i)$ is added to $u_i.\Omega$.
\item Let $t_C = max(t_A,t_B)$.
\end{itemize}

Then, at date $t_C$, either $u_i.\Omega$ or the local topology topology of $u_i$ changes. Thus, according to rule $1$ of our algorithm, $u_i$ multicasts $\Omega' = u_i.\Omega$ at date $t_C$, with $(p,p.m_0,X_i) \in \Omega'$.

As $\zeta(e_i,t_i) \leq t_{i+1} - t_i \leq t_{i+1} - t_C$, $u_{i+1}$ receives $\Omega'$ from $u_i$ at date $t_C + \zeta(e_i,t_i) \leq t_{i+1}$. Then, according to rule $2$ of our algorithm, $(p,p.m_0, X_i \cup \{u_i\})$ is added to $u_{i+1}.\Omega$.

Thus, $\mathcal{P}_{i+1}$ is true if we take $X_{i+1} = X_i \cup \{u_i\}$.
\end{itemize}

By induction principle, $\mathcal{P}_{N}$ is true. As $u_1 = p$, $X_N = \{u_1,\dots,u_{N-1}\} = S \cup \{p\}$, and we eventually have $(p,p.m_0,S \cup \{p\}) \in q.\Omega$.

Thus, $\forall S \in \{S_1,\dots,S_n\}$, we eventually have $(p,p.m_0,S \cup \{p\}) \in q.\Omega$. Then, as $\MinCut{\{S_1,\dots,S_n\}} > k$, according to rule $3$ of our algorithm, $(p,p.m_0)$ is added to $q.Acc$.
\end{proof}

\begin{lemma}[Sufficient condition]
\label{lemsuf}
Let there be any dynamic graph.
Let $p$ and $q$ be two correct nodes, and $k$ denote the maximum number of Byzantine nodes. If $\DynMinCut{p,q} > 2k$, our algorithm ensures reliable communication from $p$ to $q$.
\end{lemma}

\begin{proof}
Let us suppose that the correct nodes follow our algorithm, as described in Section~\ref{sec_algo}. First, according to Lemma~\ref{lemsafe}, if $(p,m) \in q.Acc$, then $m = p.m_0$. Thus, when $q$ accepts a message from $p$, $p$ is necessarily the author of this message. Then, according to Lemma~\ref{lemrel}, we eventually have $(p,p.m_0) \in q.Acc$. Thus, $q$ eventually receives and accepts the message broadcast by $p$. Therefore, according to Definition~\ref{def_rel}, our algorithm ensures reliable communication from $p$ to $q$.
\end{proof}

\section{Cryptographic reliable communication}
\label{sec_algo2}

If cryptography is available, then it becomes possible to authenticate the sender of a message across multiple hops.

The setting is now the following.
Each node $p$ has a private key $priv_p$ (only known by $p$)
and a public key $pub_p$ (known by all nodes).
The node $p$ can encrypt a message $m$ with the function $crypt(priv_p,m)$.
Any node $q$ can decrypt a message from $p$ with the function $decrypt(pub_p,m)$.
This function returns NULL if the message was not correctly encrypted.
We assume that the Byzantine nodes do not know the private keys of correct nodes.

Then, we modify the previous algorithm as follows.
Initially, $u.\Omega = \{(u,crypt(priv_u,u.m_0))\}$
and $u.Acc = \{(u,u.m_0)\}$. Then, each correct node $u$ obeys to the three following rules:

\begin{enumerate}
\item Initially, and whenever $u.\Omega$ or the local topology of $u$ change: multicast $u.\Omega$.

\item Upon reception of $\Omega'$ from a neighbor node:
$u.\Omega = u.\Omega \cup \Omega'$.

\item Whenever there exists $(s,m) \in u.\Omega$ such that $m' = decrypt(pub_s,m) \neq NULL$:
append $(s,m')$ to $u.Acc$.
\end{enumerate}

\begin{theorem}
\label{mainthm2}
If cryptography is available, for a given dynamic graph, 
a $k$-Byzantine tolerant reliable communication from $p$ to $q$ is feasible if and only if $\DynMinCut{p,q} > k$.
\end{theorem}

\begin{proof}
If $\DynMinCut{p,q} \leq k$, then it is possible to cut all dynamic paths between $p$ and $q$ with Byzantine nodes. Thus, $q$ never receives any message from $p$. Thus, the condition is necessary. Now, let us show that the condition is sufficient.

First, $q$ cannot accept a message $(p,m)$ with $m \neq p.m_0$. Indeed, let us suppose the opposite.
According to step $3$ of the algorithm,
it implies that we have $(p,m') \in q.\Omega$, with
$decrypt(pub_p,m') = m$.
Implying that $m' = crypt(priv_p,m)$.
Let $v$ be the first node to have $(p,m') \in v.\Omega$.
According to steps $1$ and $2$ of the algorithm, $v$ cannot be a correct node.
Thus, $v$ is Byzantine, implying that a Byzantine node knows $priv_p$: contradiction.

Besides, if $\DynMinCut{p,q} > k$, then there exists at least one dynamic path from $p$ to $q$.
Thus, for the same argument as in Lemma~\ref{lemrel},
we eventually have $(p,crypt(priv_p,p.m_0)) \in q.\Omega$.
Thus, according to step $3$ of the algorithm, $(p,p.m_0)$ is added to $q.Acc$, and the condition is sufficient.

\end{proof}

\section{Case Studies}
\label{sec_case}

In this section, we apply our conditions for reliable communication to several case studies:
participants interacting in a conference,
robots moving on a grid and agents moving in the subway.
We show the interest of multihop reliable communication.

\subsection{A real-life dynamic network: the Infocom 2005 dataset}
\label{determ_infocom}

In this section, we consider the Infocom 2005 dataset \cite{infocom_dataset}, which is obtained in a conference scenario by iMotes capturing contacts between participants. This dataset can represent a dynamic network where each participant is a node and where each contact is a (temporal) edge.

We consider an 8-hour period during the second day of the conference. In this period, we consider the dynamic network formed by the 10 most ``sociable'' nodes (our criteria of sociability is the total number of contacts reported). We assume that at most one on these nodes may be Byzantine (that is, $k = 1$).

Let $p$ and $q$ be two correct nodes. Let us suppose that $p$ wants to transmit a message to $q$ within a period of $10$ minutes. Within $10$ minutes, three types of communication can be achieved: 

\begin{itemize}

\item \emph{Direct} communication: $p$ meets $q$ directly.

\item \emph{Non-cryptographic} communication: the condition for reliable non-cryptographic
communication (Theorem~\ref{mainthm}) is satisfied.

\item \emph{Cryptographic} communication: the condition for reliable cryptographic communication
(Theorem~\ref{mainthm2}) is satisfied.

\end{itemize}

If we want to ensure reliable communication despite one Byzantine node, the simplest strategy is to wait until $p$ meets $q$ directly. Let us show now that relaying the message is usually beneficial and that our approach realizes a significant gain of performance.

\begin{figure*}
\begin{center}
\includegraphics[width=17cm]{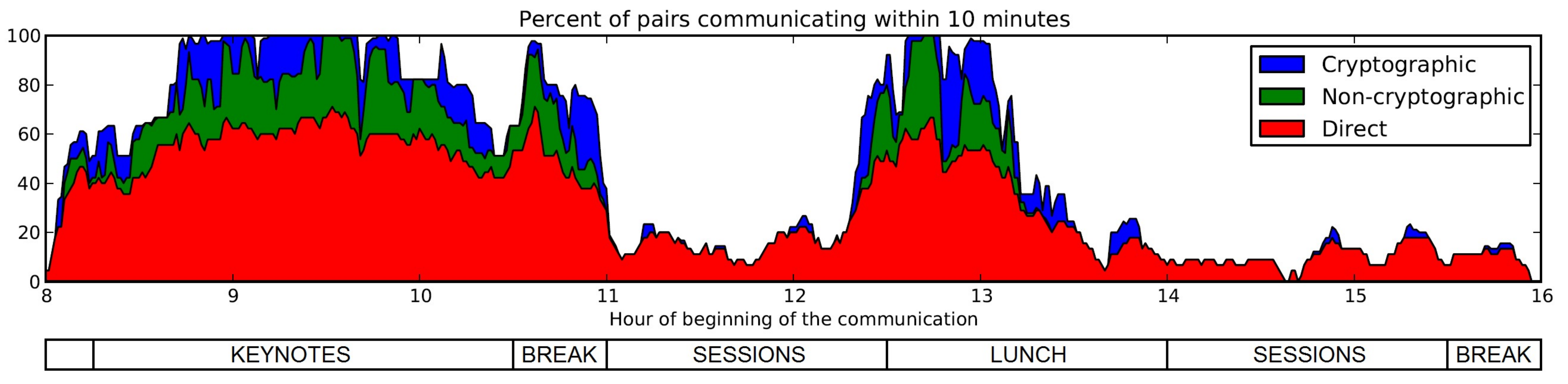}
\caption{Reliable communication between $10$ most sociable nodes of the Infocom 2005 dataset} 
\label{fig:infocom}
\end{center}
\end{figure*}

Figure~\ref{fig:infocom} represents the percentage of pairs of nodes $(p,q)$ that communicate within $10$ minutes, according to the date of beginning of the communication. We can correlate the peaks with the program of the conference: the first period corresponds to morning arrivals during the keynotes; the peak between 10:30 and 11:00 corresponds to the morning break; the peak starting at 12:30 corresponds to the end of parallel sessions and the departure for lunch.

As it turns out, many pairs of nodes are able to communicate reliably, even though they are unable to meet directly. For instance, at 9:15,
60\% of pairs of nodes meet directly, 
80\% can communicate reliably without cryptography, and 
100\% can communicate reliably with cryptography.
This means that relaying the information is actually effective and desirable.

\subsection{Probabilistic mobile robots on a grid}
\label{part_robot}

We consider a network of $10$ mobile robots that are initially randomly scattered on a $10 \times 10$ grid.

\begin{definition}[Grid]
An $N \times N$ \emph{grid} is a topology such that:
\begin{itemize}
\item Each vertex has a unique identifier $(i,j)$, with
$1 \leq i \leq N$ and $1 \leq j \leq N$.
\item Two vertices $(i_{1},j_{1})$ and $(i_{2},j_{2})$ are neighbors if and only if:
$
	\lvert j_{1}-j_{2} \rvert + \lvert i_{1}-i_{2} \rvert = 1
$
\end{itemize}
\end{definition}

At each time unit, a robot randomly moves to a neighbor vertex, or does not move (the new position is chosen uniformly at random among all possible choices). Let $position(u,t)$ be the current vertex of the robot $u$ at date $t$. We consider that two robots can communicate if and only if they are on the same vertex. Our setting induces the following dynamic graph $\mathcal{G} = (V,E,\rho,\zeta)$:
$V = \{u_1,\dots,u_{10}\}$,
$E = V \times V$,
$\rho((u,v),t) = 1$ when $position(u,t) = position(v,t)$ and
$\zeta((u,v),t) = 0$.

Let $p$ and $q$ be two correct robots, and suppose that up to $k$ other robots are Byzantine. We aim at evaluating the \emph{communication time}, that is:
the mean time to satisfy the condition for reliable communication with cryptography
(Theorem~\ref{mainthm2})
and without cryptography
(Theorem~\ref{mainthm}).
For this purpose, we ran more than 10000 simulations, and represented the results on
Figure~\ref{fig:mean_rnc} and \ref{fig:mean_rc}. Let us comment on these results.

\begin{figure}
\begin{center}
\includegraphics[width=8cm]{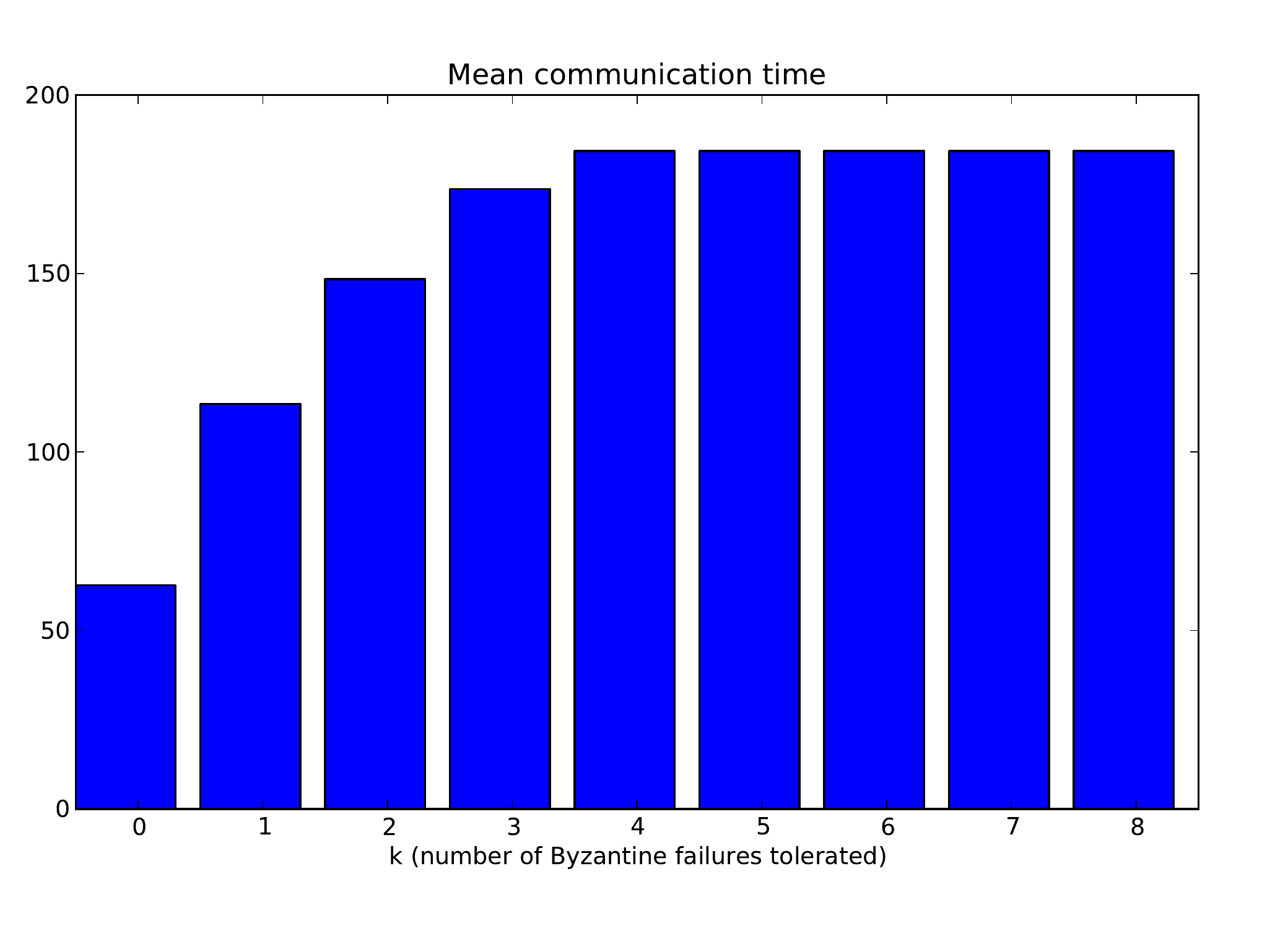}
\caption{Mean communication time without cryptography (robots)} 
\label{fig:mean_rnc}
\end{center}
\end{figure}

\begin{figure}
\begin{center}
\includegraphics[width=8cm]{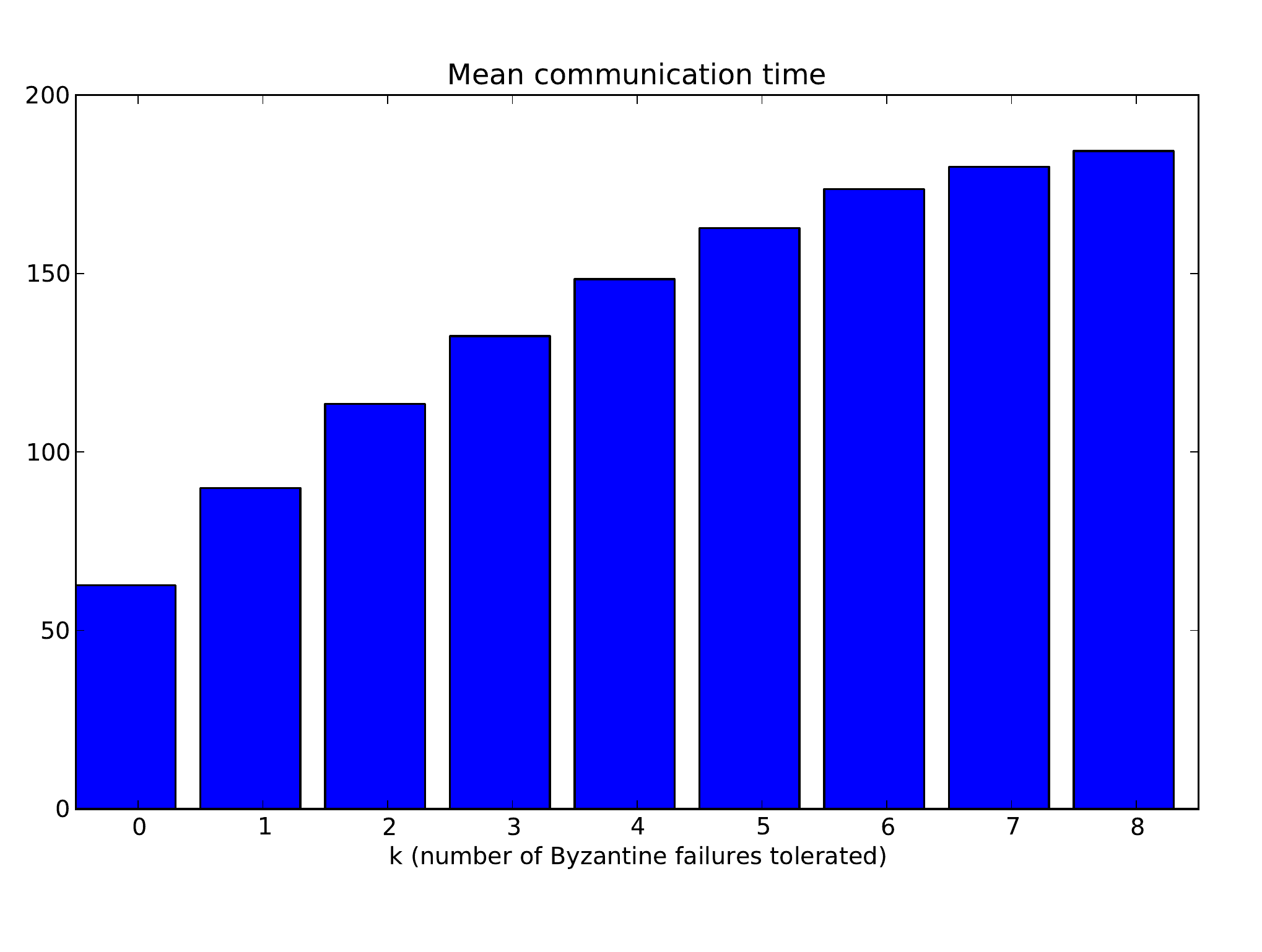}
\caption{Mean communication time with cryptography (robots)} 
\label{fig:mean_rc}
\end{center}
\end{figure}

In Figure~\ref{fig:mean_rnc}, we represented the mean communication time varying the maximal number of Byzantine failures $k$ when cryptography is available. This time increases regularly.
The case $k=8$ corresponds to the case where all the nodes (except $p$ and $q$) are Byzantine.
In this limit case, the only possibility for $p$ and $q$ to communicate is to meet directly.

In Figure~\ref{fig:mean_rnc}, we represented the case where cryptography is not available. Here, the aforementioned limit case is reached for $k = 4$, as the condition for non-cryptographic reliable communication is harder to satisfy.

As we can see, the reliable multihop communication approach can be an interesting compromise. 
For instance, let us suppose that we want to tolerate one Byzantine failure ($k = 1$).
Let us consider the mean time for $p$ and $q$ to meet directly.
If we use our algorithms, this time decreases by
38\% without cryptography,
and by 51\% with cryptography.

\subsection{Mobile agents in the Paris subway}
\label{part_subway}

We consider a dynamic network consisting of 10 mobile agents randomly moving in the Paris subway. The agents can use the classical subway lines (we exclude tramways and regional trains). Each agent is initially located at a randomly chosen junction station -- that is, a station that connects at least two lines. Then, the agent randomly chooses a neighbor junction station, waits for the next train, moves to this station, and repeats the process. We use the train schedule provided by the local subway company (\texttt{http://data.ratp.fr}). The time is given in minutes from the departure of the first train (\emph{i.e.}, around 5:30). We consider that two agents can communicate in the two following cases:

\begin{enumerate}

\item They are staying together at the same station.

\item They cross each other in trains. For instance, if at a given time, one agent is in a train moving from station $A$ to station $B$ while the other agent moves from $B$ to $A$, then we consider that they can communicate.

\end{enumerate}

\begin{figure}
\begin{center}
\includegraphics[width=8cm]{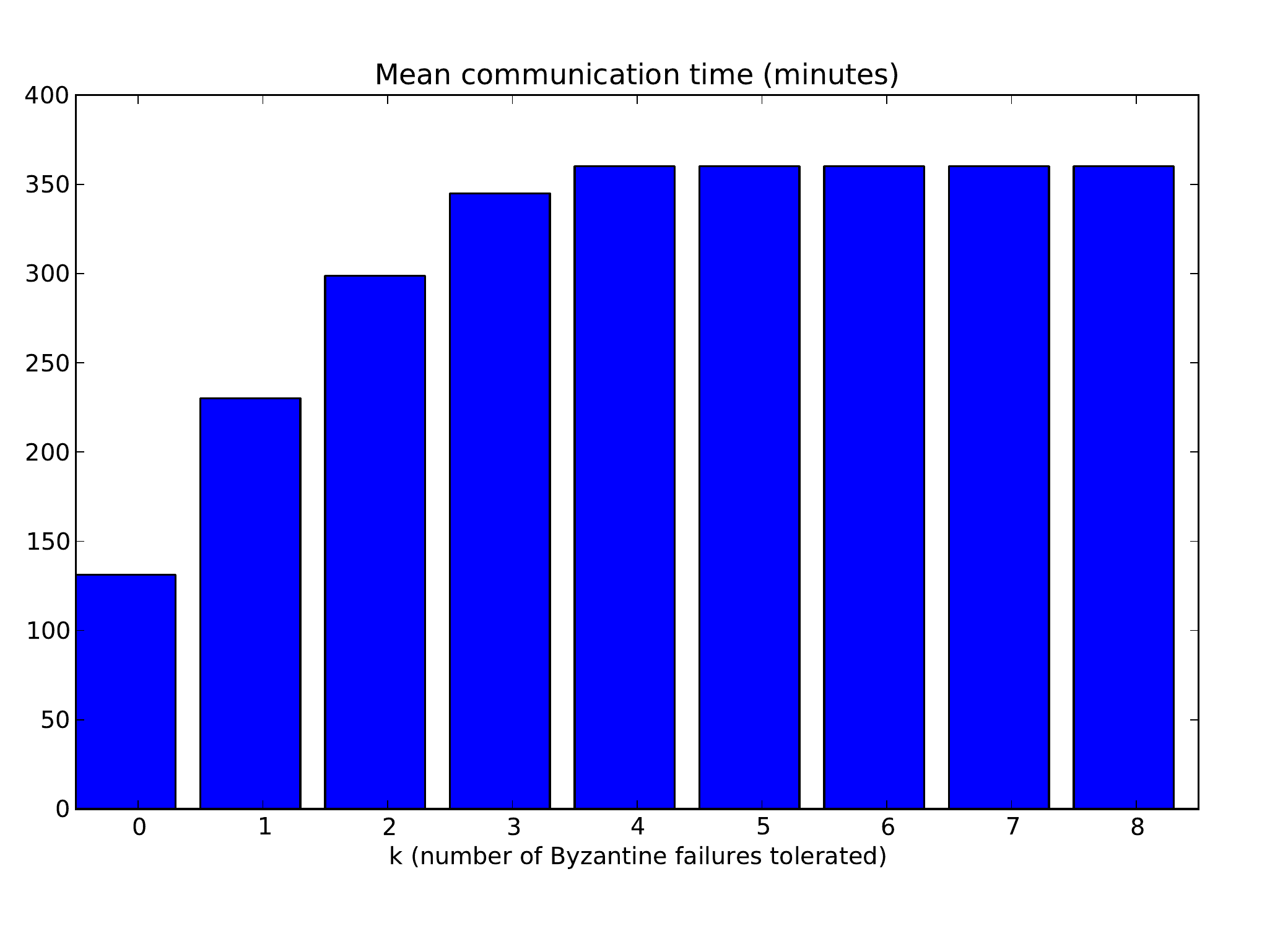}
\caption{Mean communication time without cryptography (subway)} 
\label{fig:mean_mnc}
\end{center}
\end{figure}

\begin{figure}
\begin{center}
\includegraphics[width=8cm]{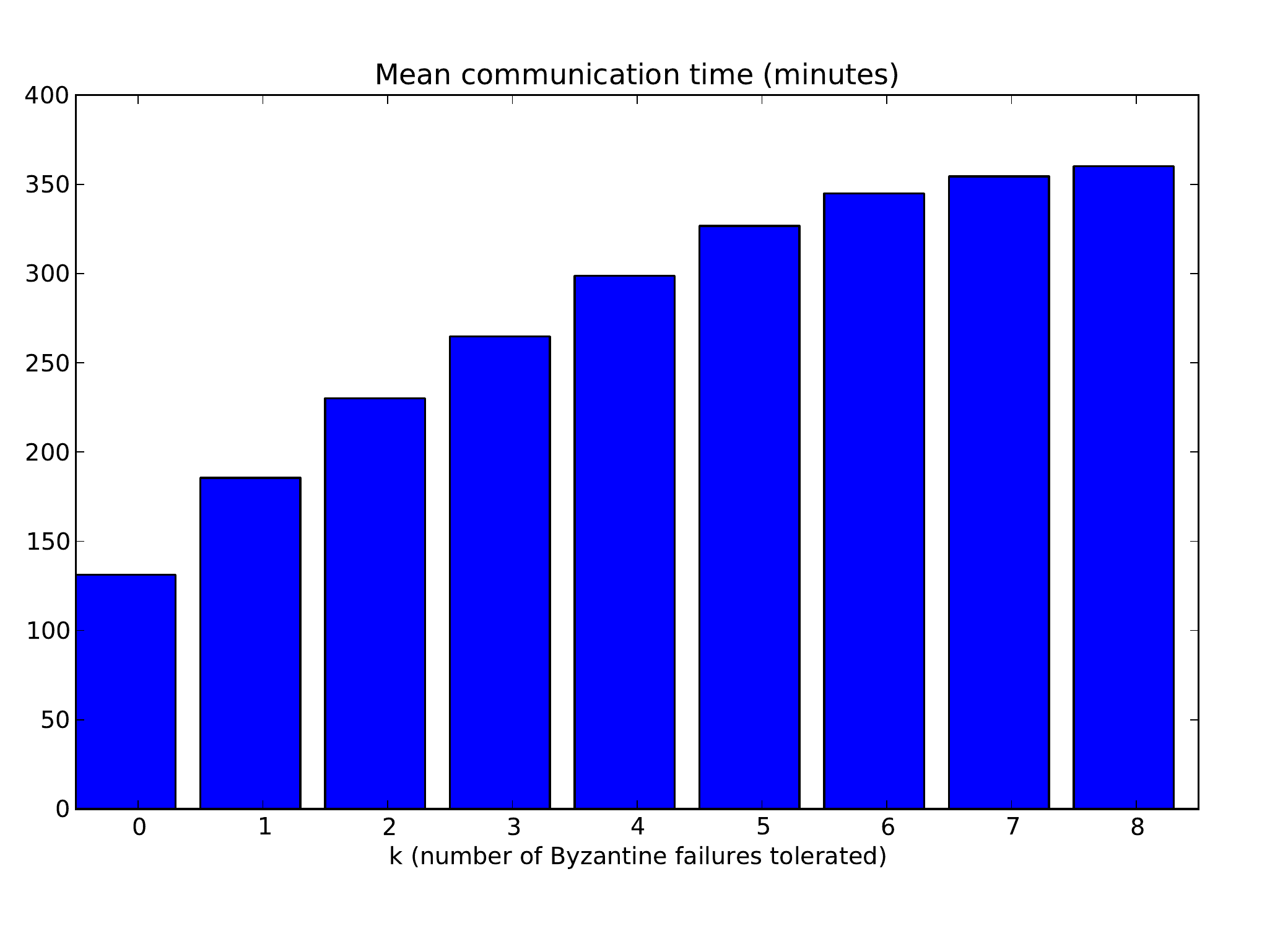}
\caption{Mean communication time with cryptography (subway)} 
\label{fig:mean_mc}
\end{center}
\end{figure}

Similarly to the previous case study,
we represented the mean communication time with and without cryptography (see Figure~\ref{fig:mean_mnc} and \ref{fig:mean_mc}). The qualitative observations are the same.

Again, let us suppose that we want to tolerate one Byzantine failure ($k = 1$).
Let us consider the mean time for $p$ and $q$ to meet directly.
If we use our algorithms, this time decreases by
36\% without cryptography,
and by 49\% with cryptography.

\section{Conclusion}
\label{sec:conclusion}

In this paper, we gave the necessary and sufficient condition for reliable communication in a dynamic multihop network that is subject to Byzantine failures. We considered both cryptographic and non-cryptographic cases, and provided algorithms to show the sufficient condition. We then demonstrated the benefits of these algorithms in several case studies.

Our experiments explicitly quantify the benefits of a cryptographic infrastructure (fewer dynamic paths are required, less computations are necessary at each node for accepting genuine messages), but additional tradeofs are worth examining. In practice, ensuring hop by hop integrity through cryptography requires every node on the (dynamic) path to collect the public key of the sender (as it is unlikely that all public keys are initially bundled into the node, for memory size reasons and inclusion/exclusion node dynamics). Actually reaching a trusted authority from a guenuinely dynamic network to obtain this public key raises both bootstrapping and performance issues.

Our result implicitly considers a worst-case placement of the Byzantine nodes, which is the classical approach when studying Byzantine failures in a distributed setting. Studying variants of the Byzantine node placement (\emph{e.g.} a random placement according to a particular distribution), and the associated necessary and sufficient condition for enabling multihop reliable communication, constitutes an interesting path for future research.

\bibliographystyle{plain}
\bibliography{biblio}

\end{document}